\documentclass[conference]{IEEEtran}
\usepackage{amsmath}
\usepackage{amssymb} 
\usepackage{amsthm} 
\usepackage{bbm}
\usepackage{mathrsfs}
\usepackage{dsfont}
\usepackage{graphicx,psfrag, color}
\usepackage{algorithm, algorithmic}
\newtheorem{theorem}{Theorem}
\newtheorem{lemma}[theorem]{Lemma}
\newtheorem{proposition}[theorem]{Proposition}

\newtheorem{ass}{Assumption}

\newcommand{\V}{\mathcal{V}}

\newcommand{\T}{\mathsf{T}}
\newcommand{\hard}{\eta}

\usepackage{epsfig} 
\usepackage{amsmath} 
\usepackage{upgreek}
\usepackage{graphicx,psfrag, color}
\usepackage{cite}
\usepackage{latexsym}
\usepackage{amssymb}
\usepackage{graphics}
\usepackage{pstricks}
\usepackage{dsfont}
\usepackage{mathrsfs}
\usepackage{bbm}
\usepackage{lettrine}%
\usepackage[T1]{fontenc}
\usepackage[english]{babel}
\usepackage[latin1]{inputenc}
\usepackage{graphicx}
\usepackage{pgf, tikz, pgflibraryshapes}
\usepackage{amsmath,amssymb}
\usepackage{colortbl}
\usepackage{multicol}
\usepackage{stackrel}

\newcommand{\cardV}{|\mathcal{V}|}

\newcommand{\fun}{\mathcal{F}}

\newcommand{\R}{\mathbb{R}}
\newcommand{\N}{\mathbb{N}}

\newcommand{\argmin}[1]{\underset{#1}{\mathrm{argmin\,}}}

\begin{document}

\title{Distributed soft thresholding\\ for sparse signal recovery}
\author{C.~Ravazzi, S.M.~Fosson, and~E.~Magli\\
Department of Electronics and Telecommunications,\\
        Politecnico di Torino, Italy 
}

\maketitle
\begin{abstract}
In this paper, we address the problem of distributed sparse recovery of signals acquired via compressed measurements in a sensor network. 
We propose a new class of distributed algorithms to solve Lasso regression problems, when the communication to a fusion center is not possible, e.g., due to communication cost or privacy reasons. 
More precisely, we introduce a distributed iterative soft thresholding algorithm (DISTA) that  consists of three steps: an averaging step, a gradient step, and a soft thresholding operation. We prove the convergence of DISTA in networks represented by  regular graphs, and we
compare it with existing methods in terms of performance, memory, and complexity.
\end{abstract}

\begin{keywords}
Distributed compressed sensing, distributed optimization, consensus algorithms, gradient-thresholding algorithms.
\end{keywords}

\section{Introduction}
Compressed sensing \cite{can06} is a new technique for nonadaptive compressed acquisition, which  takes advantage of the signal's sparsity in some domain and allows the signal recovery starting from few linear measurements. In this framework, distributed compressed sensing \cite{bar05} has emerged in the last few years with the aim of decentralizing data acquisition and processing. Most attention has been devoted to study how to perform decentralized acquisition, e.g., in sensor networks, assuming that recovery can be performed by a powerful fusion center that gathers and processes all the data. This model, however, has some significant drawbacks. First, data collection at a fusion center may be prohibitive in terms of energy utilization, particularly in large-scale networks, and may also introduce delays, which reduce the sensor network performance. Second, robustness is critical: a failure of the fusion center would stop the whole process, while some sensors' faults are generally tolerated in networks of 
considerable dimensions. Third, in several applications sensors may not be willing to convey their information to a fusion center for privacy reasons \cite{fore10}. 

In this work, we consider the problem of in-network processing and recovery in distributed compressed sensing as formulated in \cite{mat10}. 
Specifically, we assume that no fusion center is available and  we consider networks formed by sensors that can store a limited amount of information, perform a low number of operations, and communicate under some constraints. Our aim is to study how to achieve compressed acquisition and recovery leveraging on these seemingly scarce resources, with no  computational support from an external processor. The key point is to suitably exploit local communication among sensors, which allows to spread selected information through the network.  Based on this, we develop an iterative algorithm that achieves distributed recovery thanks to sensors' collaboration.

In a single sensor setting, the problem of reconstruction can be addressed in different ways. The $\ell_1$-norm minimization and Lasso are known to achieve optimal reconstruction under a sparsity model, but they are computationally expensive. E.g., the complexity of $\ell_1$-norm minimization is cubic in the length of the vector to be reconstructed. Therefore, several iterative methods have been developed, which yield suboptimal results at a fraction of the computation cost of the optimal methods \cite{for10}. On the other hand, the reconstruction problem in a distributed setting has received much less attention so far \cite{mat10}. The complexity of optimal methods is rather large also in the distributed case. This is even more important in a sensor network scenario, where the limited amount of energy and computation resources available at each node calls for algorithms with low complexity and low memory usage. This paper fills this gap, presenting a distributed iterative algorithm for compressed sensing 
reconstruction.
In particular, we propose a decentralized version of iterative thresholding methods \cite{for10}. The proposed technique consists of a gradient step that seeks to minimize the Lasso functional, a thresholding step that promotes sparsity and, as a key ingredient, a consensus step to share information among neighboring sensors.

Our aim in the design of this reconstruction algorithm is to achieve a favorable trade-off between the reconstruction performance, which should be as close as possible to that of the optimal distributed reconstruction \cite{mat10}, and the complexity and memory requirements. Memory usage is particularly critical, as microcontrollers for sensor networks applications typically have a few kB of RAM. As will be seen, the proposed algorithm is indeed only slightly suboptimal with respect to \cite{mat10}, in terms of the overall number of measurements required for a successful recovery. On the other hand, it features a much lower memory usage, making it suitable also for low-energy environments such as wireless sensor networks.

In this paper, we theoretically prove the convergence of the algorithm for networks that can be represented by regular graphs. Moreover, we numerically verify its good performance, comparing it with that of existing methods, such as distributed subgradient method (DSM, \cite{Asu10}) and alternating direction method of multipliers (ADMM,  \cite{mat10} and \cite {Boyd11}).
\section{Problem formulation}\label{sec1}
\subsection{Notation} Throughout this paper, we use the following notation. We denote column vectors with small letters, and matrices with capital letters. Given a matrix ${X}$, ${X}^{\mathsf{T}}$ denotes its transpose and $(X)_v$ (or $x_v$) denotes the $v$-th column of $X$. We consider $\R^n$ as a Banach space endowed with the following norms:
$
\|x\|_p =
\left(\sum_{i=1}^{n}|x_i|^p\right)^{1/p}$ with $p=1,2.
$
For a rectangular matrix $M\in\R^{m\times n}$, we consider  the Frobenius norm 
$
\left\|M\right\|_F=\sqrt{\sum_{i=1}^m\sum_{j=1}^nM_{ij}^2}=\sqrt{\sum_{j=1}^n\left\|(M)_{j}\right\|_2^2},
$
and the operator norm $\left\|M\right\|_2=\sup_{z\neq0}{\|Mz\|_2}/{\|z\|_2}.$
We define the sign function as $\text{sgn}(x)=1 \text{ if }x>0$, $\mathrm{sgn}(0)=0$ and
$\text{sgn}(x)=-1\ \text{otherwise}$.
If $x$ is a vector in $\R^{n}$, $\text{sgn}(x)$ is intended as a function to be applied elementwise.
A symmetric graph is a pair $\mathcal{G}=(\mathcal{V, E})$ where $\mathcal{V}$ is the set of nodes, and $\mathcal E\subseteq \mathcal{V\times V}$ is the set of edges with the property that $(i,i)\in \mathcal E$ for all $i\in\mathcal V$ and $(i,j)\in \mathcal E$ implies $(j,i)\in \mathcal E$. 
A $d$-regular graph is a graph where each node has $d$ neighbors. A matrix with non-negative elements $P$ is said to be stochastic if $\sum_{j\in \mathcal V}P_{ij}=1$ for every $i\in\mathcal V$. Equivalently, $P$ is stochastic if $P{\mathds{ 1}}={\mathds{ 1}}$. The matrix $P$ is said to be adapted to a graph $\mathcal{G}=(\mathcal{V},\mathcal{E})$ if $P_{v,w}=0$ for all $(w,v)\notin{\mathcal{E}}.$ 

\subsection{Model and assumptions}
We consider a sensor network whose topology is represented by a graph $\mathcal{G}=(\mathcal{V},\mathcal{E})$. We assume that each node $v\in\mathcal{V}$ acquires linear measurements of the form
\begin{equation}
y_v=A_vx_0+\xi_v
\end{equation}
where $x_0\in\R^n$ is a $k$-sparse signal (i.e., the number of its nonzero components is not larger than $k$), $\xi_v \in\R^m$ is an additive noise process independent from $x_0$, and $A_v\in\R^{m\times n}$  (with $n>\!\!>m$) is a random projection operator. 
If the data $(y_v,A_v)$ taken by all sensors were available at once in a single fusion center that performs joint decoding, a solution would be to solve the Lasso problem \cite{Tibshirani94regressionshrinkage}.
The Lasso refers to the minimization of the convex function $\mathcal{J}:\R^n\rightarrow\R$ defined by
\begin{equation}\label{LASSO}
\mathcal{J}(x,\lambda):=\sum_{v\in\mathcal{V}}
 \left\|y_v-{A}_vx\right\|_{2}^2+\frac{2\lambda}{\tau}\left\|x\right\|_1
\end{equation}
where $\lambda>0$ is a scalar regularization parameter that is usually chosen by cross validation \cite{Tibshirani94regressionshrinkage} and $\tau>0$. Let us denote the solution of \eqref{LASSO} as 
\begin{equation}\label{argminLasso}
\widehat{x}=\widehat{x}(\lambda)=\argmin{x\in\R^N} \mathcal{J}(x,\lambda).
\end{equation}
This optimization problem is shown to provide an approximation with 
a bounded error, which is controlled by $\lambda$ \cite{can06}.
A large amount of literature has been  devoted to developing fast algorithms for solving \eqref{LASSO} and characterizing the performance and optimality conditions. We refer to \cite{for10} for an overview of these methods.

 \subsection{Iterative soft thresholding}
A popular approach to solve \eqref{LASSO} is the iterative soft thresholding algorithm (ISTA). ISTA is based on moving at each iteration in the direction of the steepest descent followed by thresholding to promote sparsity \cite{dau04}. 

Let us collect the measurements in the vector $y=(y_1^\mathsf{T},\ldots, y_{\cardV}^\mathsf{T})^{\mathsf{T}}$ and let $A$ be the complete sensing matrix $A=(A_1^\mathsf{T},\ldots, A_{\cardV}^\mathsf{T})^{\mathsf{T}}$.
Given $x{(0)}$, iterate for $t\in \N$
\begin{align*}
x{(t+1)}&=\eta_{\lambda}(x(t)+\tau A^\mathsf{T}(y-Ax))
\end{align*}
where $\tau$ is the stepsize in the direction of the steepest descent. The operator $\eta$ is a thresholding function to be applied elementwise, i.e. $\eta_{\lambda}(x)=\text{sgn}(x)(|x|-\lambda)$ if $|x|<\lambda$ and $\eta_{\lambda}(x)=0$ otherwise.
The convergence of this algorithm was proved in \cite{dau04}, under
the assumption that $\|A\|_2^2<{1}/{\tau}.$

\section{Proposed distributed iterative soft thresholding algorithm}
In this section, we introduce our distributed iterative soft thresholding algorithm (DISTA), which has been developed following the idea of minimizing a suitable distributed version of the Lasso functional. In the next, we first discuss such a functional and then we show the algorithm.

\subsection{DISTA description}
We recast the optimization problem in \eqref{LASSO} into a separable form which facilitates distributed implementation. The goal is to split this problem  into simpler subtasks executed locally at each node.
Let us replace the global variable $x$ in \eqref{LASSO} with local variables $\{x_v\}_{v\in\mathcal{V}}$, representing estimates of $x_0$ provided by each node. While the conventional centralized Lasso problem attempts to minimize $\mathcal{J}(x,\lambda)$, we  recast the distributed problem as an iterative minimization of the functional $\mathcal{F}:\mathbb{R}^{n\times|\mathcal{V}|}\longmapsto \mathbb{R}^+$ defined as follows
\begin{align}\begin{split}\label{DLasso}
\mathcal{F}(x_1,\ldots,x_{|\mathcal{V}|}):=\sum_{v\in\mathcal{V}}&\left[q\|y_v-A_vx_v\|_2^2+\frac{2\alpha}{\tau_v|\mathcal{V}|}\|x_v\|_1\right.\\
&\left.+\frac{1-q}{\tau_v}\sum_{w\in\mathcal{V}}P_{v,w}\|\overline{x}_w-{x}_v\|_2^2\right]
\end{split}
\end{align}
where $P=[P_{v,w}]_{v,w\in\mathcal{V}}$ is a stochastic matrix adapted to the graph $\mathcal{G}$, and for some $q\in(0,1)$. 

By minimizing $\mathcal{F}$, each node seeks to recover the sparse vector $x_0$ from its own linear measurements, and to enforce agreement with the estimates calculated by other sensors in the network.
It should also be noted that  
$\mathcal{F}(\bar x,\ldots,\bar x)=q\mathcal{J}(\bar x,\lambda)$ if and only if $x_v=\overline x$, $\alpha=q\lambda$ and $\tau_v=\tau$ for all $v\in\mathcal{V}$.

Note that $q$ can be viewed as a temperature
parameter; as $q$ decreases, estimates $x_v$ associated with
adjacent nodes become increasingly correlated.
Let us denote as $\{\widehat{x}_v^{q}\}_{v\in\mathcal{V}}$ a minimizer of \eqref{DLasso}.
If $\mathcal{G}$ is connected, then we expect that $\lim_{q\rightarrow0} \widehat{x}_v^{q} = \widehat{x} $, $\forall v\in\mathcal{V}.$ This fact suggests that if $q$ is sufficiently small, then
each vector $\widehat{x}_v^{q} $ can be used as an estimate of $x_0$.

DISTA seeks to minimize \eqref{DLasso} in an iterative, distributed way. The key idea is as follows. Starting from $x_v(0)=0$ for any $v\in\V$, each node $v$ stores two messages at each time $t\in\mathbb{N}$,  ${x}_v(t)$ and $\overline{x}_v(t)$. The update is performed in an alternating fashion: at even $t$, $\overline{x}_v(t+1)$ is obtained by a weighted average of the estimates $x_w(t)$  for each $w$ communicating with $v$; at odd $t$, $x_v(t+1)$ is computed as a thresholded convex combination of a consensus and a gradient term.
More precisely, the pattern is summarized in Algorithm \ref{alg1}.
\begin{algorithm}
  \caption{DISTA}\label{alg1}
  \begin{algorithmic}
  \STATE {Given symmetric, row-stochastic matrix $P$  adapted to the graph,  $\alpha$, $\tau_v>0$, $x_v(0)=0$,  $y_v=A_vx_0$ for any $v\in\V$, iterate}
  \begin{itemize}
\item     {$t \in 2\mathbb{N}$, $v\in\V$,}
     \begin{align*}
 \overline{x}_v(t+1)&=\sum_{w\in\mathcal{V}}P_{v,w} {x}_w(t)\\
x_v(t+1)&=x_v(t)
\end{align*}
 \item {$t \in 2\mathbb{N}+1$, $v\in\V$,}
\begin{align*}
\overline{x}_v(t+1)&=\overline{x}_v(t)\\
 {x}_v(t+1)&=\eta_{\alpha}\left[(1-q)\sum_{w\in\mathcal{V}}P_{v,w} \overline{x}_w(t)\right.\\
&+q\Big( {x}_v(t)+ \tau_v A_v^{\mathsf{T}}(y_v-A_v  {x}_v(t))\Big)\Bigg]
\end{align*}
\end{itemize}
   \end{algorithmic}
\end{algorithm}
It should be noted that DISTA provides a distributed protocol: each node only needs to be aware of its neighbors and no further information about the network
topology is required. Moreover, if $\cardV=1$, DISTA coincides with ISTA. In section \ref{Tr}, we will prove its convergence and show that it minimizes $\fun$.

\subsection{Discussion and comparison with related work}
Algorithms for distributed sparse recovery (with no central processing unit) in sensor networks have been proposed in the literature in the last few
years. We distinguish two classes:
\begin{enumerate}
\item algorithms based on the decentralization of subgradient methods for convex optimization (DSM, \cite{Asu10});

\item distributed implementation of the alternate method of multipliers (ADMM, \cite{mat10,mot12, Boyd11});
\end{enumerate}

\subsubsection{DSM}\label{par_DSM}
Our proposed approach leverages distributed algorithms for multi-agent optimization that have been proposed in the literature in the last few years \cite{Asu10}.
The main goal of these algorithms is to minimize over a convex set the sum of cost functions that are convex and differentiable almost everywhere.  
It should be noted that these methods can be applied to the Lasso functional.
The memory storage requirements and the computational complexity are similar to DISTA, as it does not require to solve linear systems, to invert matrices, or to operate on the matrices $A_v$. However, we emphasize some substantial differences.

DSM is not guaranteed to converge while DISTA will be proved to converge to a minimum of  \eqref{DLasso}. The convergence can be achieved by considering the related ``stopped'' model (see page 56 in \cite{Asu10}), whereby the nodes stop computing the subgradient at some time, but they keep exchanging their information and averaging their estimates only with neighboring messages for subsequent time. However, the tricky point of such
techniques is the optimal choice of the number of iterations to stop the computation of the subgradient. Moreover, the limit point cannot be variationally characterized and depends on the time we stop the model.
In \cite{Ram10distributedstochastic}, the stepsize for the subgradient computation decreases to zero along the iterations. This choice, however, requires to fix an initial time and is not be feasible in case of time-variant input: introducing a new input would require some resynchronization. For this reason the parameters $q,\tau $ in distributed iterative thresholding algorithms are kept fixed and will be compared with DSM with constant stepsize.

\subsubsection{Consensus ADMM}
The consensus ADMM \cite{mat10, Boyd11} is a method for solving problems in which the objective and the constraints are distributed across multiple processors. The problem in \eqref{LASSO} is solved by introducing dual variables $\omega_v$ and minimizing the augmented Lagrangian in a iterative way with respect to the primal and dual variables.
 The algorithm entails the following steps for each $t\in\N$: node $v$ receives the local estimates from its neighbors, uses them to evaluate the dual price vector and the new estimate via coordinate descent and thresholding. 
The bottleneck of the consensus ADMM is the inversion of the $n\times n$ matrices $(A_v^{\mathsf{T}}A_v+\rho I)$ for some $\rho>0$. 

Although the inversion of the matrices in consensus ADMM is performed off-line, the storage of an $n\times n$ matrix may be prohibitive for a low power node with a small amount of available memory. More precisely, for the consensus ADMM each node has to store  $2+m+mn+n^2+3n$ real values. DISTA, instead, requires only  $3+m+mn+2n$ real values, that correspond to $q$, $\alpha$, $\tau_v$, $y_v$, $A_v$, $x_v(t)$, and $\overline{x}_v(t)$. Suppose that the node can store $S$ real values: for the consensus ADMM, the maximum allowed $n$ is of the order of $\sqrt{S}$, while for DISTA is of the order of $S$. For example, let us consider the implementation of DISTA on STM32F microcontrollers with Contiki operating system. Each node has 16 kB of RAM; as the static memory occupied by consensus ADMM and DISTA is almost the same, let us neglect it along with the memory used by the operating system (the total is of the order of hundreds of byte). Using a single-precision floating-point format, $2^{12}$ real values can be 
stored in 16 kB. Therefore, even assuming that the nodes take just one measurement ($m=1$), consensus ADMM can handle signals with length up to 61 samples, while DISTA up to 1364.
This clearly shows that DISTA is much more efficient in low memory devices.

\section{Main results}
\subsection{Theoretical results}\label{Tr}
In this section, we summarize our convergence analysis of DISTA.

Let $X(t)=(x_1(t),\dots,x_{\cardV}(t))$, $
\overline{X}(t)=XP^{\mathsf{T}}
$, and define the operator $\Gamma: \R^{n\times\cardV}\longmapsto \R^{n\times \cardV}
$ where
\begin{center}$
(\Gamma X)_v=\eta_{\alpha}\left[(1-q)(\overline XP^{\T})_v+q(x_v+\tau_v A_{v}^{\mathsf{T}}(y_v-A_v x_v))\right]$
\end{center}
with $v\in\mathcal V$. DISTA can be equivalently rewritten as 
$X(t+1)=\Gamma X(t)
$
with any initial condition $X(0)$.

Our goal is to find sufficient conditions that guarantee the convergence of the estimates $X(t)$ to a finite limit point, and to characterize this limit point in terms of the properties of the function $\fun$ in \eqref{DLasso}.

In our analysis, we adopt the following  assumptions. 
\begin{ass}
$\mathcal{G}$ is a $d$-regular graph, $d$ being the degree of the nodes.
\end{ass}
\begin{ass}
The nodes in $\mathcal{V}$ use uniform weights, i.e., $P_{u,v}=1/d$ if $(u,v)\in\mathcal{E}$ and zero otherwise.
\end{ass}

The following theorem ensures that, under certain conditions on the stepsize $\{\tau_v\}_{v\in\mathcal{V}}$,  the problem of minimizing the function in \eqref{DLasso} is well-posed.

\begin{theorem}[Characterization of minima]\label{fixed oints_Gamma}
If $\tau_v<\|A_v\|_2^{-2}$
for all $v\in\mathcal{V}$, the set of minimizers of $\mathcal{F} (X)$ is not empty and coincides with the set $\mathrm{Fix}(\Gamma)=\{Z\in\R^{n\times \mathcal{V}}: \Gamma Z=Z\}$. 
\end{theorem}

The proof is rather technical and is omitted for brevity. The interested reader can refer to \cite{FMR}.

Moreover,  
\begin{theorem}[Convergence]\label{teo:DISTA_convergence} If $\tau_v<\|A_v\|_2^{-2}$
for all $v\in\mathcal{V}$, DISTA produces a sequence $\{X(t)\}_{t\in\N}$ such that
$$
\lim_{t\rightarrow\infty}\left\|{X}(t)-X^{\star}\right\|_F=0
$$
where the limit point $X^{\star}$ is a fixed point of $\Gamma$.
\end{theorem}

These theorems guarantee that DISTA produces a sequence of estimates converging to a minimum of the function $\fun $ in \eqref{DLasso}. The sketch of the proof is deferred to the Appendix. 

It is worth mentioning that Theorem 2 does not imply that DISTA achieves
consensus: local estimates are very close to each other, but do not necessarily coincide at
convergence. Consensus can be obtained by letting $q$ go to zero or considering the related ``stopped'' model  whereby the nodes stop computing the subgradient at some time, but they keep exchanging their information and averaging their estimates only with neighbors messages for the rest of the time. Stopped models were introduced in \cite{Asu10}. 

\subsection{Numerical results}\label{numerical_results}
To demonstrate the good performance of DISTA, we conduct a series of experiments
for the complete graph architecture and for a variety of total number of measurements. We consider the complete topology where $P_{ij}=\frac{1}{N}$ for every $i,j=1,\dots, N$.
For a fixed $n$, we
construct random recovery scenarios for sparse vector $x_0$. For each $n$, we vary the number of measurements $m$ per node and the number of nodes in the network.
For each $(N, m, |\mathcal{V}|)$, we repeat the following procedure 50 times.
A signal is generated by choosing $k$ nonzero components  uniformly among the $n$ elements and sampling the entries from a Gaussian distribution $\mathsf{N(0,1)}$.
Matrices $(A_v)_{v\in\mathcal{V}}$ are sampled from the Gaussian ensemble with $m$ rows, $n$ columns, zero mean  and variance $\frac{1}{m}$.
We fix $n = 150$, $k=15$, $\alpha=10^{-4}$, and $\tau=0.02$. 

In the noise-free case, we show the performance of DISTA in terms of reconstruction probability as a function of the number of measurements (see Figure \ref{Fig1}).
In particular, we declare $x_0$ to be recovered if $\sum_{v\in\mathcal{V}}\|x_0-x^{\star}_v\|^2_2\big{/} (n|\mathcal{V}|)< 10^{-4}$.
\begin{figure}[t]
\includegraphics[width=0.99\columnwidth]{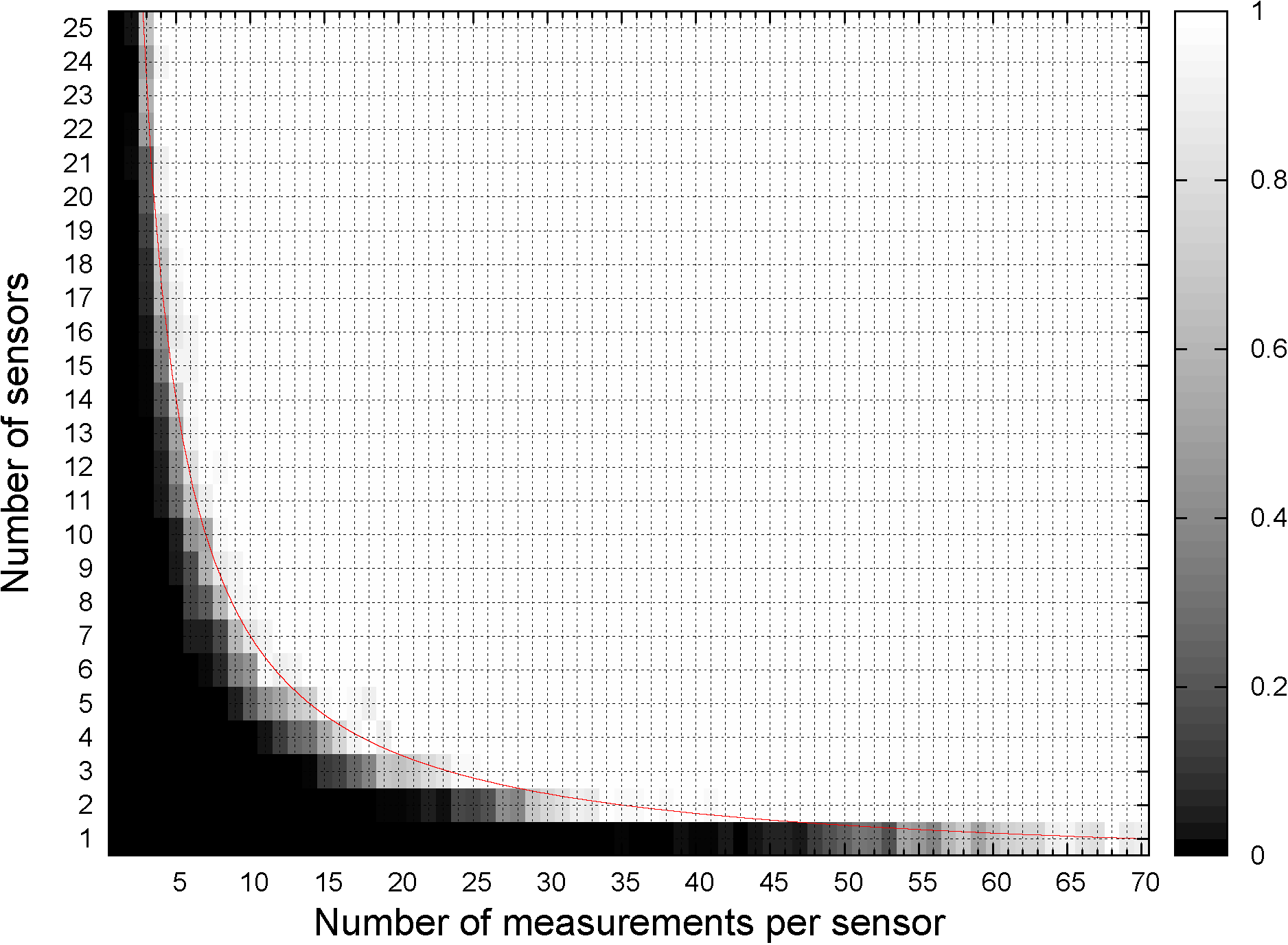}
\caption{Noise-free case: recovery probability of DISTA for a complete graph, $n=150$, $k=15$. The red curve represents $m|\V|=70.$}
\label{Fig1}
\end{figure}
The color of the cell reflects the empirical recovery rate of the 50 runs (scaled between 0 and 1). White and black respectively denote perfect recovery and failure for all experiments. It should be noted that the number of total measurements $m|\mathcal{V}|$, which are sufficient for successfully recovery, is constant: the red curve collects the points $(m,|\V|)$ such that $m|\V|=70$, which turns out to be a sufficient value to obtain  good reconstruction.

In Figure \ref{Fig2} the probability of success of DISTA, consensus ADMM, and DSM are compared as a function of the number of measurements per node. The curves are depicted for different numbers of sensors.
We  notice that the number of measurements needed for success by DISTA is smaller with respect to DSM. On the other hand, DISTA has performance close to the optimal ADMM: we obtain an almost perfect match in the curves obtained in the same scenario. In \cite{FMR}, we have compared also the time of convergence: as known, DSM has problems of slowness \cite{mat10} and thousands of steps are not sufficient to converge. DISTA, instead, is significantly faster, hence feasible. It does not reach the quickness of consensus ADMM, which however has the price of the inversion of a $n\times n$ matrix to start the algorithm. 
The interested reader can refer to \cite{FMR} for a test of the algorithm with a real dataset.

Finally, let us consider the noisy case. In Figure \ref{Fig3}, the mean square error
{\small{$$\textsf{MSE}=\frac{\sum_{v\in\mathcal{V}}\|{x_0}-x^{\star}_v\|^2_2 }{n|\mathcal{V}|},$$}}averaged over 50 runs is plotted as a function of the signal-to-noise ratio {\small{$$\textsf{SNR}=\frac{\mathbb{E}\left[\sum_{v\in\mathcal{V}}\|y_v\|_2^2\right]}{\mathbb{E}\left[\sum_{v\in\mathcal{V}}\| \xi_v\|_2^2\right]}.$$ }} The number of sensors is $|\mathcal{V}|=10$. 
The graph shows that DISTA performs better then DSM, even at larger compression level: taking $m=8$ is sufficient for DISTA to obtain a MSE lower than that obtained by DSM with $m=12$ measurements. Notice that this is the best performance that can be obtained by  DSM in this setting, that is, even without compression we do not see any improvement. On the other hand,  DISTA with $m=12$ is worse than consensus ADMM with same $m$ or with  $m=8$, but it is better than consensus ADMM with $m=6$ for sufficiently large SNR. In conclusion, DISTA can achieve the optimal performance of consensus ADMM at the price of a smaller compression level, which is not achievable by DSM.

\begin{figure}[t]
\includegraphics[width=0.99\columnwidth]{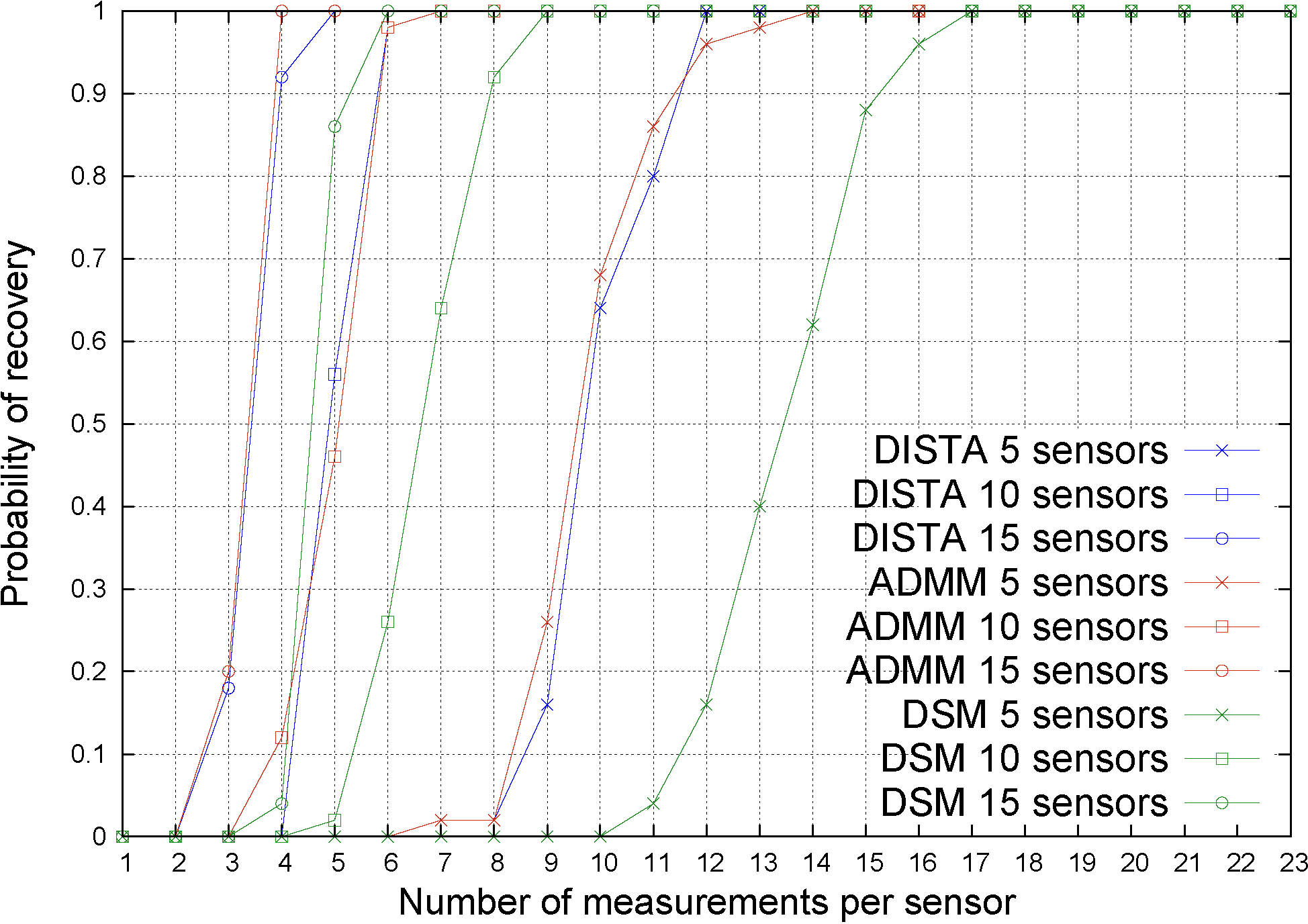}
\caption{Noise-free case:  DISTA vs DSM and consensus ADMM, complete graph, $n=150$, $k=15$.}
\label{Fig2}
\vspace{0.4cm}
\includegraphics[width=0.99\columnwidth]{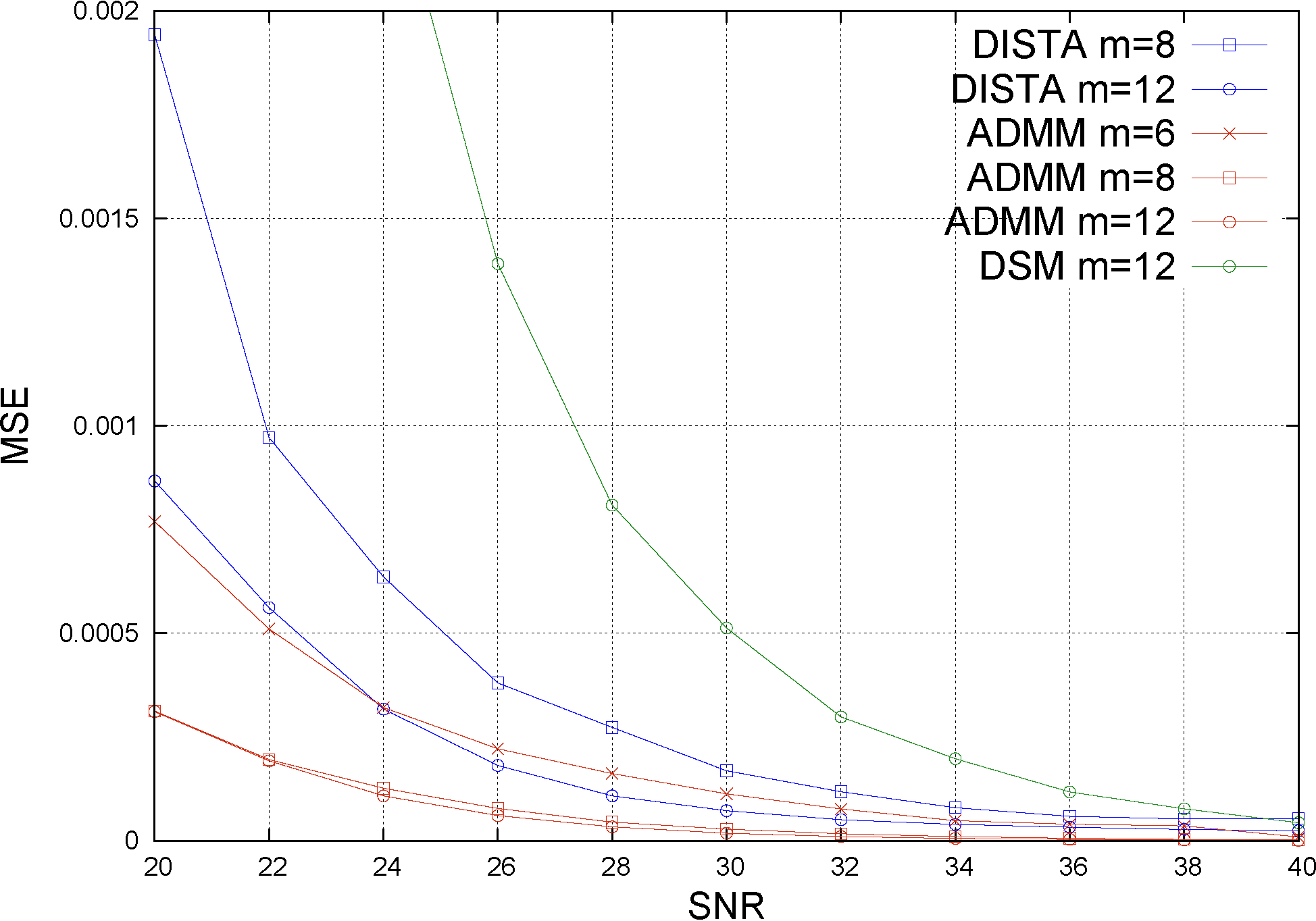}
\caption{Noise case: DISTA vs DSM and consensus ADMM, complete graph,  $n=150$, $k=15$, $|\mathcal{V}|=10$.}
\label{Fig3}
\end{figure}

\section{Concluding remarks}

The problem of distributed estimation of
sparse signals from compressed measurements in sensor networks with limited communication
capability has been studied. We have proposed the first iterative algorithm for distributed reconstruction, based on the iterative soft thresholding principle. This algorithm has low complexity and memory requirements, making it suitable for low energy scenarios such as wireless sensor networks. Numerical results show that the proposed algorithm outperforms existing distributed schemes in terms of memory and complexity, and is almost as good as the consensus ADMM method. We have also provided proof of convergence in the case of regular graphs.

\section{Acknowledgment}
This work is supported by the European Research Council under the European Community's Seventh Framework Programme (FP7/2007-2013) / ERC Grant agreement n.279848.

\bibliographystyle{IEEEbib}
\bibliography{DITA}

\appendix
\section{Proof of Theorem 2}
We provide a sketch of the proof  that the sequence of the  $\{X(t)\}_{t\in\N}$ converges to a fixed point of $\Gamma$, applying the Opial's Theorem to the operator $\Gamma$ (see Theorem 2).

\begin{theorem}[Opial's theorem  \cite{opi67}]
Let $T$ be an operator from a finite-dimensional space $\mathbb{S}$ to itself that satisfies the following conditions:
\begin{enumerate}
\item  $T$ is asymptotically regular (i.e., for any $x\in \mathbb{S}$, and for $t\in\N$, $\left\|T^{t+1} x-T^t x\right\|_2\to 0$ as $t\to\infty$);
\item $T$ is nonexpansive (i.e., $\left\|T x-T z\right\|_2\leq \left\| x-z\right\|_2$ for any $x,z \in \mathbb{S}$);
\item $\mathrm{Fix}(T)\neq \emptyset$,  $\mathrm{Fix}(T)$ being the set of fixed point of $T$.
\end{enumerate}
Then, for any $x\in\mathbb{S}$, the sequence $\{T^t(x)\}_{t\in\N}$ converges weakly to a fixed point of $T$.
\end{theorem}
It should be noticed that in $\R^{n}$ the weak convergence coincides with the strong convergence.

Let us now prove that $\Gamma$ satisfies the Opial's conditions.  
Instead of optimizing \eqref{DLasso}, let us introduce a surrogate objective function:
\begin{equation}\label{surro}
\begin{split}
& \fun ^{\mathcal{S}}(X, C, B):=\sum_{v\in\mathcal V}\bigg(q\left\|A_v x_v-y_v\right\|_2^2+\frac{2\alpha}{\tau_v}\left\|x_v\right\|_{1}\\
&~~~~~~~~+\frac{1-q}{d\tau_v}\sum_{w \in \mathcal N_v}\left\|x_v-c_w\right\|_2^2 +\frac{q}{\tau_v}\left\|x_v-b_v\right\|_2^2\\
&~~~~~~~~-q\left\|A_v (x_v- b_v)\right\|_2^{2}\bigg)
\end{split}
\end{equation} 
where $C=(c_1,\dots,c_{\cardV})\in\R^{n\times \cardV}$, $B=(b_1,\dots,b_{\cardV})\in\R^{n\times \cardV}$. 
It should be noted that, defining $\overline{X}=XP^{\mathsf{T}}$, 
$$
\fun^{\mathcal{S}}(X,\overline{X},X)= \fun (X).
$$
If we suppose $\alpha>0$ and $\tau_v<\|A_v\|_2^{-2}$ for all $v\in\V$, then $\fun^{\mathcal{S}}$  is a majorization of $ \fun $, and minimizing $\fun^{\mathcal{S}}$  leads to a majorization-minimization (MM) algorithm.
The optimization of \eqref{surro} can be computed by minimizing with respect to each $x_v$ separately.

\begin{proposition}\label{eta}The following fact holds:
{\small{\begin{align*}
\argmin{x_v\in\R^n}  \fun ^{\mathcal{S}} (X,C,B)=\hard_{\alpha}\left[(1-q)\overline{c}_v+q b_v+q\tau A_v^{\mathsf{T}}(y_v-A_vb_v)\right]
\end{align*}}}
where $\overline{c}_v=\frac{1}{d}\sum_{w\in\mathcal N_v}c_w$.
\end{proposition}

\begin{proposition}\label{medie}  If $\tau_v<\left\|A_v\right\|_2^{-2}$ for all $v\in\V$ the following facts hold:
\begin{align}
\argmin{c_v\in\R^n} \fun ^{\mathcal{S}}(X,C,B)&=\frac{1}{d}\sum_{w\in\mathcal{N}_v}{x}_w,\\
\label{stationary point}
\argmin{b_v\in\R^n} \fun ^{\mathcal{S}} (X,C,B)&=x_v.
\end{align}
\end{proposition}

In few words, in DISTA the vectors $x_v$ and $\overline{x}_v$ are updated in an alternating fashion, separating the minimization of the consensus part from the regularized least square function in \eqref{DLasso}. 

We now prove that $\Gamma$ is asymptotically regular, i.e., that $X(t+1)-X(t)\rightarrow 0$ for $t\rightarrow\infty.$ In particular, this property guarantees the numerical convergence of the algorithm.

\begin{lemma}\label{decrescenza}If $\tau_v<\left\|A_v\right\|_2^{-2}$ for all $v\in\mathcal{V}$, then the sequence $\{ \fun (X(t))\}_{t\in\N}$ is non increasing and admits the limit.
\end{lemma}
\begin{proof}
The function is lower bounded ($\mathcal{F}(X)\geq 0$) and the sequence $\{\mathcal{F}_p(X(t)\}_{t\in\mathbb{N}}$ is decreasing and therefore admits the limit. 

From Lemma \ref{eta} and Lemma \ref{medie} we obtain the following inequalities:
\begin{align*}
 \fun (X(t+1))&\leq \fun ^{\mathcal{S}}(X(t+1),\overline{X}(t+1),X(t))\\
&\leq \fun ^{\mathcal{S}}(X(t+1),\overline{X}(t),X(t))\\
&\leq  \fun ^{\mathcal{S}}(X(t),\overline{X}(t),X(t))= \fun (X(t)).
\end{align*}
\end{proof}

\begin{proposition}\label{succ_ite}
For any $\tau_v\leq\min_{v\in\mathcal{V}}\left\|A_v\right\|_2^{-2}$ the sequence $\{X(t)\}_{t\in\N}$ is bounded and
$$
\lim_{t\rightarrow+\infty}\left\|X(t+1)-X(t)\right\|_F^2=0.
$$
\end{proposition}
\begin{proof}
By Lemma \ref{eta} we obtain 
\begin{align*}
&\mathcal{F}(X(t))-\mathcal{F}(X(t+1))\\
&\quad \geq  \fun ^{\mathcal{S}}(X(t+1),\overline{X}(t+1),X(t))\\
&\quad- \fun ^{\mathcal{S}}(X(t+1),\overline{X}(t+1),X(t+1))\\
&\quad\geq \frac{q}{\tau_v}\sum_{v\in\mathcal{V}}(x_v(t+1)-x_v(t))^{\mathsf{T}}M_v(x_v(t+1)-x_v(t))\geq 0.
\end{align*}
Notice that the last expression is nonnegative as $M_v=I-\tau_v A_v^{\mathsf{T}}A_v$ are positive definite for all $v\in\V$.

As $ \fun ^{\mathcal{S}}(X(t))- \fun ^{\mathcal{S}}(X(t+1)\to 0$ we thus conclude that $\left\|x_v(t+1)-x_v(t)\right\|_2^2\to 0$ for any $v\in\mathcal V$ and
$$\lim_{t\rightarrow+\infty}\left\|X(t+1)-X(t)\right\|_F^2=0.$$
\end{proof}

 We now prove that $\Gamma$ is nonexpansive.

\begin{lemma}\label{nonexp}
For any $\tau\leq\min_{v\in\mathcal{V}}\left\|A_v\right\|_2^{-2}$, $\Gamma$ is nonexpansive.
\end{lemma}
\begin{proof}
Since $\eta_{\alpha}$ is nonexpansive, for any $X, Z \in \R^{n\times\cardV}$,
\begin{align*}
&\left\|(\Gamma X)_v-(\Gamma Z)_v\right\|^2_2\\
&\leq\left\|(1-q)(\overline{\overline{ x}}_v-\overline{\overline{ z}}_v)+q(I-\tau A_{v}^{\mathsf{T}}A_v)( x_v- z_v)\right\|^2_2\\
&\leq\left[(1-q)\left\|\overline{\overline{ x}}_v-\overline{\overline{ z}}_v\right\|_2
+q\left\|I-\tau A_{v}^{\mathsf{T}}A_v\right\|_2\left\| x_v- z_v\right\|_2\right]^2.\\
\end{align*}
Notice that $I-\tau A_{v}^{\mathsf{T}}A_v$ always has the eigenvalue 1 with algebraic multiplicity $n-m$, as the rank of $A_v$ is $m$. Moreover, if $\tau<\left\|A_v\right\|_2^{-2}$, $I-\tau A_{v}^{\mathsf{T}}A_v$ is positive definite and its spectral radius is 1. Since $I-\tau A_{v}^{\mathsf{T}}A_v$ is a symmetric matrix, we then have  $\left\|I-\tau A_{v}^{\mathsf{T}}A_v\right\|_2=1$. Thus, applying the triangular inequality,
\begin{align*}
&\left\|(\Gamma X)_v-(\Gamma Z)_v\right\|^2_2\leq\left[(1-q)\left\|(\overline{\overline{ x}}_v-\overline{\overline{ z}}_v)\right\|_2
+q\left\| x_v- z_v\right\|_2\right]^2\\ 
&\leq \frac{(1-q)^2}{d^4}\left(\sum_{w\in\mathcal N_v}\sum_{w'\in\mathcal N_w}\left\| x_{w'}-z_{w'}\right\|_2\right)^2
+q^2\left\| x_v- z_v\right\|_2^2\\
&~~~+ \frac{2(1-q)q}{d^2}\sum_{w\in\mathcal N_v}\sum_{w'\in\mathcal N_w}\left\| x_{w'}-z_{w'}\right\|_2\left\| x_v- z_v\right\|_2.
\end{align*}

Applying the Cauchy-Schwarz inequality, we obtain 
\begin{align*}
&\left\|(\Gamma X)_v-(\Gamma Z)_v\right\|^2_2\\
&\leq \frac{(1-q)^2}{d^2}\sum_{w\in\mathcal N_v}\sum_{w'\in\mathcal N_w}\left\| x_{w'}-z_{w'}\right\|_2^2
+q^2\left\| x_v- z_v\right\|_2^2\\
&~~~+ \frac{2(1-q)q}{d^2}\sum_{w\in\mathcal N_v}\sum_{w'\in\mathcal N_w}\left\| x_{w'}-z_{w'}\right\|_2\left\| x_v- z_v\right\|_2.
\end{align*}
Finally, summing over all $v\in\mathcal V$ and considering that  $2\left\| x_{w'}- z_{w'}\right\|_2\left\| x_v- z_v\right\|_2\leq \left\| x_{w'}- z_{w'}\right\|^2_2+\left\| x_v- z_v\right\|^2_2$,
\begin{align*}
&\left\|(\Gamma X)-(\Gamma Z)\right\|^2_F=\sum_{v\in\mathcal{V}}\left\|(\Gamma X)_v-(\Gamma Z)_v\right\|^2_2\\
&\leq (1-q)^2\left\| X- Z\right\|_F^2+q^2\left\| X- Z\right\|_F^2+2(1-q)q \left\| X- Z\right\|_F^2\\
&\leq \left\| X- Z\right\|_F^2.
\end{align*}
\end{proof}

\textbf{Proof of Theorem 2}  Given the numerical convergence (proved in Proposition 7), the existence of fixed points (guaranteed by Theorem 1), and the nonexpansivity (Lemma \ref{nonexp}) of the operator $\Gamma$, the assertion follows from a direct application of the Opial's theorem.\qed

\end{document}